\documentclass[twoside,10.5pt]{article}
\usepackage{mathrsfs}
\usepackage{pifont}
\usepackage{amsmath}
\usepackage{amsthm}
\usepackage{txfonts}
\usepackage{geometry}
\usepackage{latexsym}
\usepackage{amssymb}
\usepackage{graphicx}
\usepackage{geometry}
\usepackage{xcolor} 
 \newcommand{\eps}{\varepsilon}

 \newcommand{\R}{\mathbb{R}}
 
 \newcommand{\N}{\mathbb{N}}

 \newcommand{\prob}{\mathbb{P}}

 \newcommand{\me}{\mathbb{E}}
 
 \renewcommand{\P}{\mathbb{P}}
 
 \newcommand{\skric}{{\mathcal C}}
 \newcommand{\skrid}{{\mathcal D}}

 \newcommand{\skrim}{{\mathcal M}}

 \newcommand{\skrin}{{\mathcal N}}

 \newcommand{\skrix}{{\mathcal X}}

 \newcommand{\sfrac}[2]{\mbox{$\frac{#1}{#2}$}}

\def\1{{\mathchoice {1\mskip-4mu\mathrm l}      
{1\mskip-4mu\mathrm l} {1\mskip-4.5mu\mathrm l} {1\mskip-5mu\mathrm
l}}}

\newcommand{\eq}{\begin{equation}}
\newcommand{\en}{\end{equation}}

\newenvironment{Proof}
{\vskip0.1cm\noindent{\bf Proof. }{\hspace*{0.3cm}}}%
{\nopagebreak {\hspace*{\fill}\rule{2mm}{2mm}}\\ }

{\nopagebreak {\hspace*{\fill}\rule{2mm}{2mm}}\\ }

\renewcommand{\subsection}{\secdef \subsct\sbsect}
\newcommand{\subsct}[2][default]{\refstepcounter{subsection}
\vspace{0.15cm} {\flushleft\bf
\arabic{section}.\arabic{subsection}~\bf #1  }
\nopagebreak\nopagebreak}
\newcommand{\sbsect}[1]{\vspace{0.1cm}\noindent
{\bf #1}\vspace{0.1cm}}

\newtheorem{theorem}{Theorem}[section]
\newtheorem{lemma}[theorem]{Lemma}

\newtheoremstyle{thm}{1.5ex}{1.5ex}{\itshape\rmfamily}{}
{\bfseries\rmfamily}{}{2ex}{}

\newtheoremstyle{rem}{1.3ex}{1.3ex}{\rmfamily}{}
{\itshape\rmfamily}{}{1.5ex}{} \theoremstyle{rem}

\refstepcounter{subsection}

\def\thebibliography#1{\section*{References}
  \list%
  {\arabic{enumi}.}
    {\settowidth\labelwidth{[#1]}\leftmargin\labelwidth
    \advance\leftmargin\labelsep
    \parsep0pt\itemsep0pt
    \usecounter{enumi}}
    \def\newblock{\hskip .11em plus .33em minus .07em}
    \sloppy                   
    \sfcode`\.=1000\relax}

\begin{document}
\begin{center}
{\LARGE{ Large deviations, Basic information theorem for fitness
preferential
attachment random networks }}\\[20pt]
\end{center}

\centerline{\sc{By K. Doku-Amponsah, F. O. Mettle and  T.
Narh-Ansah}}
\renewcommand{\thefootnote}{}
\footnote{\textit{Mathematics Subject Classification :}     4A15;
    94A24;
    60F10;
    05C80.}

\footnote{University  of   Ghana, Statistics Department, P. O. Box
lg 115, Accra,kdoku@ug.edu.gh}
\renewcommand{\thefootnote}{}

\textbf{Abstract}

For fitness preferential  attachment   random networks, we define
the \emph{ empirical degree and pair measure}, which counts  the
number of vertices of a given degree and the  number of edges with
given fits, and the sample path \emph{empirical degree
distribution}.
 For  the empirical degree and pair distribution for the
 fitness preferential  attachment  random networks, we  find   a large  deviation upper  bound. From this  result we obtain a   weak  law  of  large
 numbers for the empirical degree and pair distribution, and  the   basic information theorem or an asymptotic equipartition
 property
 for fitness
 preferential  attachment  random networks.

\textbf{Keywords:} Large deviation upper bound, relative entropy,
random network, random  tree, random coloured graph, typed graph,
asymptotic equipartition property.

\textbf{1. Introduction}

This  paper establishes an asymptotic  equipartition property (AEP)
for  fitness preferential  attachment (P.A) random  networks. The
AEP is an important characteristics used often in information theory
to partition output samples of a stochastic data source. See,
example  (Doku-Amponsah, 2010)  and the references therein for
similar result for  networked  datasets modelled as   coloured
random graphs or random fields.

In  the past three decades technological advances in  the Social
Sciences, Web Science and related fields have yielded large amounts
of diverse networked datasets which are  best described in terms of
the preferential attachment graphs. Example the WWW, consisting of
over 800 million documents (vertices) and a large number  of links (
edges) pointing from one document to another, is best model by
preferential attachment graphs. See,( Lawrence and  Giles,
1998,1999). In order to transmit or compress  datasets from this
random network source, one require efficient coding schemes and
approximate pattern matching algorithms, and the AEP for P.A
networks   play a key role in this regard, example by providing
bounds on the possible performance of these schemes or algorithms.

P.A   models  can be easily defined and modified, and can therefore
be calibrated to serve as models for social networks and the web
graph. These graphs model fairly well the dynamics  of the
occurrence  of power law degree distributions in large networks. See
( Barabasi and Albert, 1999).

The main  ideal  behind the  P.A  models  is  that  growing networks
are constructed  by  adding  nodes successively. If  a  new  node
appears, it  gets a  fit or  colour or symbol or  spin according  to
some law $\mu$ on a finite  alphabet and it is linked by an edge to
one or more existing node(s) with a probability proportional  to
function of their degree and fits. The dynamics of    the graph  is
completely determined by the function $f$ known as the attachment
rule.

There  are three  regime  of  P.A  graphs. Namely, for  \emph{linear
regime}: $f(k)\approx k$,\emph{sublinear regime}:$f(k)\le k$ and
\emph{superlinear regime}:$f(k)\ge k.$ Several results  about  the
asymptotic behaviour  of  these   graphs have been established
recently.

Few  large  deviation results for  P.A  model have so far been
found. In article ( Choi et. al, 2011), P.A schemes where the
selection mechanism is possibly time-dependent are considered, and
an in infinite dimensional large deviation principle for the sample
path evolution of the empirical degree distribution is found by
Dupuis-Ellis type methods.

(Dereich and  Morters, 2009) studied a dynamic model of random
networks, where new vertices are connected to old ones with a
probability proportional to a sublinear function of their degree.
For this model of  random  networks, they  obtained a strong limit
law for the empirical degree distribution. Results on the temporal
evolution of the degrees of individual vertices via large and
moderate deviation principles  were also found.

(Bryc et. al, 2009)  found  the large  deviation principle and
related results  for a class of Markov chains associated to the
`leaves'in P.A model  of  random  graphs using both analytic and
Dupuis--Ellis-type path arguments.

In this  article,  we    prove  a large deviation upper  bound  for
the empirical  degree and pair distribution, and use  it  to  find
an AEP for for P.A models of random graphs in the linear regime $f.$
i.e. $f(k)\approx k.$ Our proofs use  the techniques  of exponential
change change-of-measure for random graphs, see example (Dembo et.
al, 2003), (Doku-Amponsah, 2006), (Doku-Amponsah and Morters, 2010),or ( Doku-Amponsah, 2010). 

To  be  specific, we  prove a  large deviation upper,  see
 Theorem~\ref{PA3}, for the empirical degree and  pair distribution  of the
fitness P.A model  of  random  graphs. For a given, empirical degree
and  pair distribution we prove  from the large deviation upper
bound a weak law  of  large  numbers,see
 Theorem~\ref{PA5}. And  from this weak  law  of  large  numbers  we find  the AEP for  a  networked structure  datasets model,
see Theorem~\ref{PA41},  as  a fitness P.A model of random graphs.

\textbf{2. Large deviation  upper  bound for P.A random graphs}

We write $\skrin=\N\cup\{0\}.$ Given a weight function
$f:\skrin\times\skrix\to[0,\,\infty]$ and a probability law  $\mu$
on  finite alphabet $\skrix,$ we define coloured(fitness) P.A random
network as follows:
\begin{itemize}
\item Assign vertex $n=1$ (the root of the network) colour $X(n)$
according to $\mu:\skrix\to[0,\,1].$
\item If a new vertex $n$ is introduced, it gets colour $X(n)$ independently according $\mu,$
\item it connects to vertices  $v_n\in\{\,1,\ldots,n-1\,\}$ independently with
probability proportional to $$ f(N(v_n),A(n)),$$  where
$A(n)=\big(X(v_n),X(n)\big)$
 and $N(m)$ is the in-degree of vertex $m.$ We consider
$\big\{(N(v_n),A(n)):\,n=1,2,3,\ldots\big\}$ under the joint law of
colour and tree. Denote by $X$ a typed tree and by $X(i)$ colour of
vertex $i.$
\end{itemize}

We  write $\skrix^*=\skrix\times\skrix.$ In  this  paper, we shall
restrict ourself to functions of the form
$$f(k,a)=\gamma(a)k+\beta(a),$$ where
$\gamma:\skrix^* \to(0,\,\infty]$, $\beta:\skrix^*\to[0,\,\infty].$
We assume
\begin{equation}\label{equ.one} \gamma(a)+\beta(a):=c,\, \mbox{ for
all $a\in\skrix.$}
\end{equation}
and  that the function $f$ satisfy the following  weak preference
condition:
\begin{equation}\label{persistentEqPA}
\inf_{a\in\skrix}\sum_{k=0}^{\infty}\sfrac{1}{f(k,a)}=\infty.
\end{equation}

 Let
$N^{(m)}(i)$ be the degree of vertex $i$ at time $m$ and observe
that at time $n,$ the law of the fitness  P.A graph is given by
$$\begin{aligned}
\P_{f}^{(n)}(X)=\prod_{m=1}^{n}\mu(X(m))\times &
\prod_{m=2}^{n}\frac{f(N^{(m)}(j_m),
\,A(m))}{\sum_{i=1}^{m-1}f(N^{(m)}(i),\,A(i)).}
\end{aligned}
$$

 For every $X,$ we define empirical degree  and  pair  measure
measure $M_{X}$ on $\skrin\times\skrix^*$ by
$$M_{X}(k,\,a)=\frac{1}{n-1}\sum_{m=1}^{n-1}\delta_{(N^{(m)}(j_m),A(m))}(k,\,a).$$

We  write  $\ell_m(a)=\big\{ j_m\in\{1,2,3,...,m-1\}:
x(j_m)=a_1,\,x(m)=a_2\big\}$  and  for every $m=2,3,4,...,n-1$ we
define a probability measure on $\skrin\times\skrix^{*}$ by

$$ L_{\sfrac{m}{n}}^{X}(k, a)=\frac{1}{m-1}\sum_{j=1}^{m-1}\delta_{N^{(m)}(j)} (k)\1_{\{j\in\ell_m(A(m)\}}\otimes\delta_{A(m)}(a), $$

where
$$\begin{aligned}\label{Eqdef1}
\1_{\{j\in\ell_m(b)\}}\otimes\delta_{b}(a)=\left\{
\begin{array}{ll}  \1_{\{j\in\ell_m(b)\}}\,
& \mbox { if $b=a,$  }\\
0 & \mbox{otherwise.}
\end{array}\right.
\end{aligned}$$

and notice, $$L_{1}^{X}(k , a)=M_{X}(k,a).$$  We denote by
$\skrim(\skrix)$ the space of probability measures on $\skrix$
equipped with the weak topology and $\skrim(\skrin\times\skrix^{*})$
the space of probability measures on $\skrin\times\skrix^*,$
equipped with the topology generated by total variation
metric.$$\|\pi-\hat{\pi}\|:=\frac{1}{2}\sum_{(k,a)\in
\skrin\times\skrix^{*}}\|\pi(k,a)-\hat{\pi}(k,a)\|.$$

 We are now in the
position to state our  large deviation upper bound  for the fitness
P.A
 model  of  random  graphs.  We write
 $\hat{\omega}(k \,| \,a):=\1-\sum_{j=0}^{k}\omega(k\,|\,a)$  and  state  our
 large  deviation  upper bound  for   the empirical  pair  measure.

\begin{theorem}\label{PA3}
Suppose $X$ is coloured P.A random graph with colour law
$\mu:\skrix\to(0,1]$ and linear weight function
$f:\skrin\times\skrix^*\to[0,\infty].$ Then, for  any  close
$\Gamma\subset\skrim(\skrin\times\skrix^*),$

$$
\limsup_{n\to\infty}\frac{1}{n}\log\P\Big\{M_X\in\Gamma\big\}\le-\inf_{\omega\in\Gamma}J(\omega),$$

$$J(\omega)=H\Big(\omega_{2,1}\,\|\,\mu\Big)
+\sum_{a\in\skrix}\omega_{2,1}(a)H\Big(\omega(\cdot|a)\,\|\,\sfrac{c}{f(\cdot,\,a)}\otimes\hat{\omega}(\cdot|\,a)\Big),$$

where $\omega_{2,1}$  is  the $\skrix-$ marginal of the probability
measure $\omega_2$ and
$$\sfrac{c}{f(\cdot,\,a)}\otimes\hat{\omega}(\cdot|\,a)(k)
=\frac{c}{f(k,\,a)}\hat{\omega}(k\,|\,a).$$

\end{theorem}
Observe  that   $J(\omega)=0$   if  and  only  if
$\omega(k,a)=\frac{c\omega_2(a)}{f(k,\,a)}\big(\1-\sum_{j=0}^{k}\omega(k\,|\,a)\big),$
and  hence solving  recursively  for   $\omega(\cdot\,|\,a)$  we get
\begin{equation}\label{Rudas}
\omega(k\,|a)=\pi_f(k\,|a):=\frac{c}{c+f(k,a)}\prod_{i=0}^{k-1}\frac{f(i,a)}{c+f(i,a)}.
\end{equation}

Here we  remark  that  conditions \eqref{equ.one} and
\eqref{persistentEqPA}  are necessary  for $\pi_f(\cdot\,|a)$ to  be
a probability  measure on $\skrin$.  See (Dereich  and Morters,
2009, ~p.~13). Note, if  $f(k,a)=w(k)$  then \eqref{Rudas} concise
with  the asymptotic degree distribution  of random  trees and
general branching processes found in (Rudas et. al, 2008).

\textbf{3. Basic  Information Theorem for fitness P.A random
networks}

Our  main  theorem  is  the AEP for networked datasets modelled as
fitness P.A graph.  In  this section, we state the AEP for networked
data structure described by fitness P.A graphs. By $P$ we denote the
(probability) law of a fitness P.A  graph with $n$ vertices. Thus we
write
$$P(x):=\P\Big\{X=x\Big\},\,\mbox{  for each fitness P.A graph
$x.$}$$

\begin{theorem}\label{PA4}
Suppose $X$ is coloured P.A random graph with colour law
$\mu:\skrix\to(0,1]$ and linear weight function
$f:\skrin\times\skrix^*\to[0,\infty].$ Then, for  any $\eps>0,$

$$
\lim_{n\to\infty}\P\Big\{\Big|\sfrac{1}{n}\log
P(X)-\sum_{a_1\in\skrix}\mu(a_1)\log
\mu(a_1)-\sum_{(k,\,a)\in\skrin\times\skrix^*}\mu\otimes
\mu(a)\pi_f(k\,|a)\log f(k,\,a)/c\Big|\ge \eps\Big\}=0.$$
\end{theorem}

In other words, in order to transmit a coloured P.A graph in the
given regime one needs with high probability, about
$$\frac{n}{\log 2}\Big[\sum_{a_1\in\skrix}\mu(a_1)\log
\mu(a_1)+\sum_{(k,\,a)\in\skrin\times\skrix^*}\mu\otimes
\mu(a)\pi_f(k\,|a)\log f(k,\,a)/c\Big]\,\mbox{ bits}.$$

\textbf{4. Proof  of  THeorem~\ref{PA3}}

\emph{4.1  Dynamics of  the  sample path empirical degree
distribution}

Denote by $\skrid([0,1],\R)$ the  space  of  right continuous left
limited(cadlag) paths  from $[0,1]$  to  $\R.$   We define  the
sample  path space
$$\begin{aligned}
&\skrid_{\skrim}:=D([0,1]:\skrim(\skrin\times\skrix))\\
&=\Big\{\mbox{the set of  all $\nu:[0,1]
\mapsto\skrim(\skrin\times\skrix)$ such that
$\nu(k,a)\in\skrid([0,1],\R)$ for  all  $k\ge 0, a\in\skrix$  and
$\langle \nu\rangle=1$}\Big\}
\end{aligned}
$$ and  endow  it with the  topology of  uniform convergence
associated with  the norm
 $$\|\nu-\hat{\nu}\|:=\sup_{t\in[0,1]}\|\nu_t-\hat{\nu}_t\|.$$

For  any $\nu\in\skrid_{\skrim}$ we  write
$\nu_t(k\,|a):=\sfrac{\nu_t(k,\,a)}{\sum_{k=0}^{\infty}\nu_t(k,\,a)},$
for  all $t\in[0,1]$  and $(k,a)\in\skrin\times\skrix.$  Write
$\dot{\nu}_{t}:=\frac{d\nu_t}{dt}$ for  the  time  derivative  of
 the  measure $\nu_t$ and we associate with each path $\nu\in\skrid_{\skrim}$  the relaxed
measure on $[0,1]\times(\skrin\times\skrix)$
 $$\bar{\nu}(dk,dt|a)=\nu_{t}(dk|a)dt.$$

We call $\nu\in\skrid_{\skrim}$ absolutely continuous if for  each
$k\in\N$, there  exists $\dot{\nu}(k|a)$ such that
$$\nu_1(k|a)-\nu_0(k|a)=\int_{0}^{1}\dot{\nu}_s(k|a)ds.$$
For  each absolutely continuous path $\nu$ ,  we define
$\nu^{\nu}(\cdot|a),$ \,$\bar{\nu}(\cdot,\cdot|a)$- almost
everywhere by

$$\nu_t^{\nu}(k|a):=-\sum_{i=0}^{k}\dot{\nu}_{t}(i|a).$$  By $\nu^{\nu}\ll\nu $
we mean $\nu$  is  absolutely  continuous. We  write
$$\skrid_{\skrim_n(\skrin\times\skrix)}:=\Big\{\nu\in\skrid_{\skrim(\skrin\times\skrix)}:\,
([nt]-1)\nu_{[nt]/n}\in\N,\,\forall t\in[0,1)\Big\}.$$

Note  that  the  measure  $L_{{[nt]}/{n}}^{X},$ for  $t\in[0,1)$ is
deterministic  and  its  distribution  is  degenerate  at  some
$\nu_{{[nt]}/{n}},$  for $t\in[0,1)$  converging  to $\nu_t,$
$t\in[0,1).$

\emph{4.2 Exponential Change-of- Measure}

 Throughout  the
remaining  part  of  this  paper  we  shall  assume $\nu_t(k|a)\le
\nu_t^{\nu}(k|a),$ for  all  $t\in [0,\,1].$

Let  $\tilde{g}: \N\times\skrix\to \R$, and write $\displaystyle
\lim_{n\to\infty}L_{\sfrac{[nt]}{n}}:=\nu_t\in\skrid_{\skrim},$ we
define the function $U_{\tilde{g}}:[0,\,1]\times \skrix\to \R$ by
$$ U_{\tilde{g}}^{(n)}\otimes\nu(a,t)=\log \frac{\langle
\sfrac{e^{\tilde{g}(\cdot,\,a)}}{f},\,\nu_{\sfrac{[nt]}{n}}(\cdot|a)\rangle}{\langle
f,\,\nu_{\sfrac{[nt]}{n}}(\cdot|a)\rangle},$$ and  note  that
$$\lim_{n\to\infty}U_{\tilde{g}}^{(n)}\otimes{\nu}(a,t)=\log\frac{\langle \sfrac{e^{\tilde{g}(\cdot,\,a)}}{f},\,\nu_{t}(\cdot|a)\rangle}{\langle
f,\,\nu_{t}(\cdot|a)\rangle}=:U_{\tilde{g}}\otimes{\nu}(a,t).$$ We
use $\tilde{g}$  to  define  a  new fitness P.A random  graph  as
follows:
\begin{itemize}
\item At  time  $m=1$ assign the root $m$ of the network fit $X(m)$  according  to  the  law
$\tilde{\mu}$ given  by
$$\tilde{\mu}(a_1)=e^{\tilde{h}(a_1)-U(\tilde{h})}\mu(a_1).$$
\item For  any  other time $m=2,3,4,....n$ new node $m$  which appear  gets fit $X(m)$ according  to  the  fit  law $\tilde{\mu}.$ It connects  to node $v_m,$ independently with probability proportional to 

$$\tilde{f}(N^{(m)}(v_m),A(m))=
\frac{1}{f(N^{(m)}(v_m),A(m))}e^{\tilde{g}(N^{(m)}(v_m),A(m))}.$$

\end{itemize}

We denote by $\P_{\tilde{f},n}$ the law of the new fitness P.A graph
and observe that it is absolute continuous with respect to
$\P_{{f},n},$ as for fitness graph $X$ we have that

\begin{align}
\sfrac{d\P_{\tilde{f},n}}{d\P_{{f},n}}(X)&=\prod_{m=1}^{n}\sfrac{\tilde{\mu}(X(m)}{\mu(X(m)}\times\sfrac{\prod_{m=1}^{n-1}\tilde{f}(N^{(m)}(j_m),\,X(m))}{\prod_{m=2}^{n-1}\sum_{i=1}^{m-1}\tilde{f}(N^{(m)}(i),\,X(m))}
\times\sfrac{\prod_{m=2}^{n-1}\sum_{i=1}^{m-1}{f}(N^{(m)}(i),\,X(m))}{\prod_{m=1}^{n-1}{f}(N^{(m)}(j_m),\,X(m))}\\
&=e^{(n-1)\Big\langle
\tilde{h}-U(\tilde{h}),\,M_X\Big\rangle+(n-1)\Big\langle\tilde{g}-2\log
f,\,M_{X}\Big\rangle-(n-1)\Big\langle U_{\tilde{g}}\otimes
L,\,M_{X}\otimes id \Big\rangle},\label{EqPA4}
\end{align}\\

where  $id$ is the  identity function from $[0,1]$  to $[0,1].$ The
following Lemma will be used to establish the upper bound in a
variational formulation.

\begin{lemma}\label{PA5}
For every $\theta>0$ there exits a
 compact set $K_{\theta}\subset\skrim(\skrix^*)$ such that
\begin{equation}\label{EqPA6}
\limsup_{n\to\infty}\sfrac{1}{n}\log\prob_{f,n}\Big\{ M_X\not\in K\,
\big |(L_{[nt]/n}=\nu_{[nt]/n,\,\forall t\in(0,1]})\Big\} \leq -
\theta.
\end{equation}
\end{lemma}
\begin{proof}
Let  $1\ge\delta>0,$  and  $l\in N.$  We  choose  $k(l,\delta)\in
\N$ large  enough such  that,  for  large $n,$  we have
                        $$\sum_{i=1}^{[nt]-1}e^{l^2\1_{\{N^{([nt])}(i)>k(l,\delta)\}}}\sfrac{f(N^{([nt])}(i),a)}{
                        c([nt]-1)}\le 2e^{\delta},\,\mbox{ for all $a\in \skrix$ and  for all  $t.$}$$
Now using  Chebyschev's  inequality we  have
$$
\begin{aligned}
\P_{f,n}\Big\{M_{X}(N^{([nt])}&>k(l,\delta))\ge l^{-1},\,
L_{[nt]/n}=\nu_{[nt]/n,\,\forall t\in(0,1]}\Big\}\\
&\le
e^{-nl}\me\Big\{e^{\sum_{m=1}^{n-1}l^2\1_{\{N^{(m)}(j_m)>k(l,\delta)\}}},\,L_{\sfrac{m}{n}}=\nu_{\sfrac{[m]}{n}},\, m=2,3,4,...,n-1\Big\}\\
&=
e^{-nl}\prod_{m=2}^{n}\me\Big\{e^{l^2\1_{\{N^{(m)}(j_m)>k(l,\delta)\}}},\,L_{\sfrac{m}{n}}=\nu_{\sfrac{[m]}{n}}\,\Big\}\\
&\le e^{-nl}\Big [\sup_{a\in\skrix} \sup_{t\ge
0}\Big(\sum_{i=1}^{[nt]-1}e^{l^2}\1_{\{N^{([nt])}(i)>k(l,\delta)\}}\sfrac{f(N^{([nt])}(i),a)}{([nt]-1)
                        \Big\langle f,\,\nu_{\sfrac{[nt]}{n}}(\cdot|a)\Big\rangle}\Big)\Big]^{n}\\
&= e^{-nl}\Big [\sup_{a\in\skrix} \sup_{t\ge
0}(\sum_{i=1}^{[nt]-1}e^{l^2}\1_{\{N^{([nt])}(i)>k(l,\delta)\}}\sfrac{f(N^{([nt])}(i),a)}{
                        c([nt]-1)})\Big]^{n}\\
 &\le  e^{-nl}\times(2e^{\delta})^{n}\\
 &=e^{n(l-\delta-\log2)}
\end{aligned}
$$

Now given  $\theta$  we choose  $M>\theta+\delta+\log2$  and  define
the set $$\Gamma_{\delta, \theta}:=\big\{\nu:\nu(N>k(l,\delta))<
l^{-1}, l\ge M\big\}$$

As  $\big\{N\le k(l,\delta)\big\}$  is  pre-compact,
$\Gamma_{\delta}$   is  compact  in  the weak topology by  prokohov
criterion. Moreover

$$\P_{f,n}\Big\{M_{X}\not\in K_{\theta}\,\big| (L_{[nt]/n}=\nu_{[nt]/n,\,\forall t\in(0,1]})\Big\}\le
\sfrac{1}{1-e^{-1}}\sfrac{e^{-\theta}}{\P\Big\{L_{[nt]/n}=\nu_{[nt]/n,\,\forall
t\in(0,1]}\Big\}}=\sfrac{1}{1-e^{-1}}e^{-\theta}.$$

Now  letting  $K_{\theta}$   be  the  closure of
$\cap_{1\ge\delta>0}\Gamma_{\delta,\theta}$ and  taking  limit as
$n$ approaches $\infty$  we  have \eqref{EqPA6}  which  ends the
proof the  Lemma.

\end{proof}

\emph{4.3 Proof  of  Theorem~\ref{PA3}}

We derive the upper bound in a variational formulation. To do this,
we  denote by $\skric_1$ the space  of  all  functions  on $\skrix$
and  by $\skric_2$  the space of  all  bounded continuous functions
on $\skrin\times\skrix^*.$

We define on  the space of probability measures
$\skrim(\skrin\times\skrix)$ the function $\hat{K}$ given by

\begin{equation}\begin{aligned}\label{ratePA7}
\hat{K}_{\nu}(\omega)=\sup_{\tilde{g}\in
\skric_2,\tilde{h}\in\skric_1}\Big\{\int
(\tilde{h}-U(\tilde{h}))\omega_{2,1}(da_1)+\int
\tilde{g}(k,a)\omega(dk,da)&-2\int\log f(k,a)\omega(dk,da)\\
&-\int U_{\tilde{g}}\otimes{\nu}(a,t)\omega_2(da)\otimes dt\Big\}.
\end{aligned}
\end{equation}

\begin{lemma}\label{PA8}
For every close set $ F\subset\skrim(\skrin\times\skrix)$ we have
\begin{equation}\label{EqPA9}
\limsup_{n\to\infty}\sfrac{1}{n}\log\prob_{f,n}\Big\{ M_X\in F \Big
|(L_{[nt]/n}=\nu_{[nt]/n,\,\forall t\in(0,1]})\Big\} \leq
-\inf_{\omega\in F }\hat{K}_{\nu}(\omega)
\end{equation}
\end{lemma}
\begin{Proof}
 We let $\tilde{h}\in\skric_1$,  $\tilde{g}\in\skric_2$ and use  the Jensen's inequality to obtain
$$\begin{aligned}
&e^{(\sup_{a_1}\tilde{h}(a)-\inf_{a_1}\tilde{h}(a_1))}\le\int
e^{\tilde{h}(X(n))-U(\tilde{h})}d\tilde{\P}_{f,n}\\
&=\me\Big\{e^{(n-1)\Big[\Big\langle
\tilde{h}-U(\tilde{h}),\,M_X\Big\rangle+\Big\langle\tilde{g}-2\log
f,\,M_{X}\Big\rangle-\Big\langle U_{\tilde{g}}\otimes
L,\,M_{X}\otimes id
\Big\rangle\Big]},(L_{[nt]/n}=\nu_{[nt]/n,\,\forall
t\in(0,1]})\Big\}.
\end{aligned}$$
This yields the inequality
\begin{equation}\label{EqPA10}
\limsup_{n\to\infty}\sfrac{1}{n}\log\me\Big\{e^{(n-1)\Big[\Big\langle
\tilde{h}-U(\tilde{h}),\,M_X\Big\rangle+\Big\langle\tilde{g}-2\log
f,\,M_{X}\Big\rangle-\Big\langle U_{\tilde{g}}\otimes
L,\,M_{X}\otimes id \Big\rangle\Big]}\Big
|(L_{[nt]/n}=\nu_{[nt]/n,\,\forall t\in(0,1]})\Big\}=0.
\end{equation}
Given $\eps>0,$ define $\hat{K}_{\eps,\nu}$ by
$\hat{K}_{\nu,\eps}(\omega)=\min\big\{\hat{K}_{\nu}(\omega),{\eps}^{-1}\big\}-\eps.$
For $\omega\in F$ we fix $\tilde{h}\in\skric_1$ and
$\tilde{g}\in\skric_2$ such that
$$\langle
\tilde{h}-U(\tilde{h}),\,\omega_{2,1}\rangle+\langle\tilde{g}-2\log
f,\,\omega \rangle-\langle
 U_{\tilde{g}}^{\nu},\,\omega\otimes id\rangle\ge\hat{K}_{\nu,\eps}(\omega).$$

Now, because the function $\tilde{g}$ is bounded, we can find open
neighbourhood $B_{\omega}$  of $\omega$, such that
\begin{equation}\label{EqPA11}
\inf_{\tilde{\omega}\in B_{\omega}}\Big\{ \langle
\tilde{h}-U(\tilde{h}),\,\omega_{2,1}\rangle+\langle\tilde{g}-2\log
f,\,\omega \rangle-\langle
 U_{\tilde{g}}^{\nu},\,\omega\otimes id\rangle\,\Big\}
\ge\hat{K}_{\nu,\eps}(\omega)-\eps.
\end{equation}
Take  $\delta=\eps,$ apply the Chebyshev's inequality to
\eqref{EqPA11} and use \eqref{EqPA10}  to  get
\begin{equation}\label{EqPA12}
\begin{aligned}
&\limsup_{n\to\infty}\sfrac{1}{n}\log\prob_{f,n}\Big\{M_X\in B_{\omega}\Big |(L_{[nt]/n}=\nu_{[nt]/n,\,\forall t\in(0,1]})\Big\}\\
&\leq\limsup\sfrac{1}{n}\log\me\Big\{e^{(n-1)\Big[\Big\langle
\tilde{h}-U(\tilde{h}),\,M_X\Big\rangle+\Big\langle\tilde{g}-2\log
f,\,M_{X}\Big\rangle-\Big\langle U_{\tilde{g}}\otimes
L,\,M_{X}\otimes id \Big\rangle\Big]}\Big |(L_{[nt]/n}=\nu_{[nt]/n,\,\forall t\in(0,1]})\Big\}\\
&\qquad\qquad-\hat{K}_{\nu,\eps}(\omega)+\eps\\
&\leq-\hat{K}_{\nu,\eps}(\omega)+2\eps
\end{aligned}
\end{equation}
Using Lemma~\ref{PA5} with $\theta=\eps^{-1} $ we may choose the
compact set $G_\eps$ such that
$$
\limsup_{n\to\infty}\sfrac{1}{n}\log\prob_{f,n}\Big\{ M_{X}\not\in
G_\eps\Big |(L_{[nt]/n}=\nu_{[nt]/n},\,\forall t\in(0,1])\Big\} \leq
- \eps^{-1}. $$ Now, the set $F\cap G_{\eps}$ is compact and
therefore we may be covered by finitely many sets
$B_{\omega_1},\,\ldots,\,B_{\omega_r}$, with $\omega_i\in F$ , for
$i=1,\,\ldots,\,r.$ Hence, we have that

$$\begin{aligned}\prob_{f,n}\Big\{
M_X\in F \Big |L=(L_{[nt]/n}=\nu_{[nt]/n},\,\forall
t\in(0,1])\Big\}&\leq \sum_{i=1}^{r}\prob\Big\{ M_X\in
B_{\omega_i}\Big |(L_{[nt]/n}=\nu_{[nt]/n,\,\forall
t\in(0,1]})\Big\}\\
&\qquad\qquad+\prob\Big\{ M_X\not\in G_\eps \Big
|(L_{[nt]/n}=\nu_{[nt]/n,\,\forall t\in(0,1]})\Big\}.\end{aligned}
$$

Next we use \eqref{EqPA12} to obtain for small enough $\eps>0,$
$$\begin{aligned}
\limsup_{n\to\infty}&\sfrac{1}{n}\log\prob_{f,n}\Big\{ M_X\in F\Big
|(L_{[nt]/n}=\nu_{[nt]/n,\,\forall
t\in(0,1]})\Big\}\\
&\leq
\max_{i=1}^{r}\limsup_{n\to\infty}\sfrac{1}{n}\log\prob_{f,n}\Big\{
M_X\in B_{\omega_i}\Big |(L_{[nt]/n}=\nu_{[nt]/n,\,\forall
t\in(0,1]})\Big\}-\eps^{-1}\le-\hat{K}_{\nu,\eps}(\omega)+2\eps
\end{aligned}$$

Taking $\eps\downarrow 0$ we get the desire statement.
\end{Proof}

We show that the function $\hat{K}_{\nu}(\omega)$ in Lemma~\ref{PA8}
may be replaced by the good rate function
$$K_{\nu}(\omega)=H\Big(\omega_{2,1}\,\|\,\mu\Big)
+\sum_{a\in\skrix}\omega_{2}(a)H\Big(\omega(\cdot|a)\,\|\,\frac{c}{f(\cdot,\,a)}\otimes\int_{0}^{1}\nu_t(\cdot|a)dt\Big).$$

\begin{lemma}\label{PA13}
 For every
 $\nu\in\skrid_{\skrim}$ we have
 that $ \hat{K}_{\nu}(\omega)\ge K_{\nu}(\omega).$ Moveover, the
  function $K_{\nu}$ is good rate  function and lower semi-continuous on
  $\skrim(\skrin\times\skrix).$
 \end{lemma}
 \begin{Proof}Suppose $\nu_1=\omega$.Then, using the Jensen's inequality,
 by  our assumption  \eqref{equ.one} and  the variational characterization of entropy we  have
 $$\begin{aligned}
H\Big(\omega_{2,1}\,\|\,\mu\Big)
=\sup_{\tilde{h}}\Big\{\int\tilde{h}(a_1)\omega_{2,1}(da_1)-\log\int
e^{\tilde{h}(a_1)}\mu(da_{1})\Big\}
\end{aligned}$$

$$\begin{aligned}
&\sum_{a\in\skrix}\omega_{2}(a)H\Big(\omega(\cdot|a)\,\|\,\frac{c}{f(\cdot,,a)}\otimes\int_{0}^{1}\nu_t(\cdot|a)dt\Big)\\
&=\sup_{\tilde{g}}\Big\{\int
\tilde{g}(k,a)\omega(dk,da)-\log\int\int c
\sfrac{e^{\tilde{g}(k,\,a)}}{f(k,\,a)}\omega_2(da)\int\nu_{t}(dk|a)dt
\, \Big\}\\
&\le\sup_{\tilde{g}}\Big\{\int \tilde{g}(k,a)\omega(dk,da)-2\log
c-\int\int\log\Big(\int
\sfrac{e^{\tilde{g}(k,\,a)}}{cf(k,\,a)}\nu_{t}(dk|a)\Big)\omega_2(da)dt
\, \Big\}\\
&=\sup_{\tilde{g}}\Big\{\int \tilde{g}(k,a)\omega(dk,da)-2\log \int
f(k,a)\omega(dk,da)-\int\int\log\Big(\int
\sfrac{e^{\tilde{g}(k,\,a)}}{cf(k,\,a)}\nu_{t}(dk|a)\Big)\omega_2(da)dt
\, \Big\}\\
&\le\sup_{\tilde{g}}\Big\{\int \tilde{g}(k,a)\omega(dk,da)-2\int\log
f(k,a)\omega(dk,da)-\int\int\log\Big(\sfrac{\langle
\sfrac{e^{\tilde{g}}}{f},\,\nu_{t}(\cdot|a)\rangle}{\langle
f,\,\nu_{t}(\cdot|a)\rangle}\Big)\omega_2(da)dt
\, \Big\}\\
&=\sup_{\tilde{g}}\Big\{\int \tilde{g}(k,a)\omega(dk,da)-2\int\log
f(k,a)\omega(dk,da)-\int U_{\tilde{g}}^{\nu}(a,t)\omega_2(da)\otimes
dt\Big\}\\
&=\hat{K}_{\nu}(\omega)
\end{aligned}$$

Recall  the definition of $K_{\nu}$   above and  notice, mapping
$\omega\to K_{\nu}(\omega)$ is  continuous function. Moreover, for
all $\alpha<\infty$, the level sets $\{K_{\nu}\le \alpha\}$ are
contained in the bounded set
$$\Big\{\omega\in\skrim(\skrin\times\skrix)\colon
\,\sum_{a\in\skrix}\omega_{2}(a)H\Big(\omega(\cdot|a)\,\|\,\frac{c}{f(\cdot,\,a)}
\otimes\int_{0}^{1}\nu_t(\cdot|a)dt\Big)\le\alpha\Big\}$$ and are
therefore compact. Consequently, $K_{\nu}$ is a good rate function.

\end{Proof}

 \emph{4.4 Proof  of  Theorem \ref{PA3} By  Mixing }

 To  use  the
technique of  mixing  LDP  results developed in (Biggins, 2004), we
check  the  main criteria needed for the validity of (Biggins, 2004,
Theorem~5(a))  in  the  following   Lemma. We write
$\Theta_n:=\skrid_{\skrim_{n}(\skrin\times\skrix)},$
$\Theta:=\skrid_{\skrim(\skrin\times\skrix)},$  and  define
 $$P_{f,n}(\nu_1):=\P\Big[M_X=\nu_1\,\big|\,L_{\sfrac{[nt]}{n}}^X(\cdot,a)=\nu_{\sfrac{[nt]}{n}}(\cdot,a),\,
 t\in[0,1) \mbox{ and } a\in\skrix \Big]$$
 $$P_{n}\Big(\nu_{\sfrac{[nt]}{n}},\,t\in[0,1)\Big):=\P\Big\{L_{\sfrac{[nt]}{n}}^X=\nu_{\sfrac{[nt]}{n}}\, \Big\}$$

Then,  the  joint  distribution  of $M_X$  and  $L^X$ is obtained by
the  mixture  of $P_{f,n}$ and $P_{n}$ as follows:
$$ d\tilde{P}_{f,n}(\nu,\,\nu_1):=
dP_{n}(\nu)dP_{f,n}(\nu_1).$$

\begin{lemma}\label{Mix}\,\,
  The family  of distributions  (i) $(P_{f,n},\,n\in\N)$ (ii)
$(\tilde{P}_{f,n},\,n\in\N)$ are exponentially  tight.

\end{lemma}
\begin{proof}

(i)  As  this  family distributions  obey  a  large  deviation upper
bound  with  a  good  rate  function  $K_{\nu}(\omega),$ the family
$(P_{f,n},\,n\in\N)$ is exponentially tight. See, e.g. (Dembo and
Zeitouni, 1998, Exercise~4.1.10(c)).

(ii) By  (i)  for  every  $\theta_2$  we  can  find  $K_{\theta_2},$
compact subset of $\skrid_{\skrim(\skrin\times\skrix)}$  such that,
we have
$$\limsup_{n\to \infty}\frac{1}{n}\log
P_{f,n}(K_{\theta_2}^c)\le -\theta_2.$$  Also  by Lemma~\ref{PA5},
for every $\theta_1$  we can  find $K_{\theta_1},$ compact subset of
$\skrim(\skrin\times\skrix)$ such that,  we have
$$\limsup_{n\to \infty}\frac{1}{n}\log P_{f,n}(K_{\theta_1}^c)\le
-\theta_1.$$

Take  $\theta=\min(\theta_1,\theta_2)$ and define the relatively
compact set $\Gamma_\theta$  by

$$\Gamma_{\theta}:=\Big\{(\nu_1,\nu)\in\skrim(\skrin\times\skrix)\times\skrid_{\skrim(\skrin\times\skrix)}:\,\nu_1\in
K_{\theta_1} \mbox{ and }\nu\in K_{\theta_2}\Big\}.$$  Now,  let
$\delta>0$    and notice  that,  for  sufficiently  large $n$  we
have that
$$\tilde{P}_{f,n}(\Gamma_{\theta}^c)\le \P\big\{ M_X\in K_{\theta_1}^c \big\}+ \P\big\{ L^X\in K_{\theta_2}^c
\big\}\le C(\theta)e^{-n(\theta-\delta)}.$$

Taking  limit  $n\to\infty$ followed  by   $\delta\downarrow 0$  of
above  inequality,  yields
$$\limsup_{n\to\infty}\frac{1}{n}\log\tilde{P}_{f,n}(\Gamma_{\theta}^c)\le
-\theta$$ which  proves the second  part  of  the  Lemma.
\end{proof}

Now, as  $J(\nu_1)$  is  lower  semi-continuous  by  the  continuity
of the  relative  entropies,  and  by  Lemma~\ref{Mix} the  families
of distributions (i) $(P_{f,n},\,n\in\N)$ (ii)
$(\tilde{P}_{f,n},\,n\in\N)$  are  exponentially  tight,  we have
that the latter obeys  a  large  deviation  upper bound  with  good
rate function give  by  $J(\nu_1).$  See, (Biggins, 2004,
Theorem~5(a) and proof).

We  obtain  the  form  of  the  rate  function  in Theorem~\ref{PA3}
by  noting  that
$$\int_{0}^{1}\nu_t(k\,|a)dt\le\int_{0}^{1}\nu_t^{\nu}(k\,|a)dt=-\int_{0}^{1}
\sum_{i=0}^{k}\dot{\nu}(i\,|a)=1-\sum_{i=0}^{k}\nu(i\,|a).$$

\emph{ 4.5 Proof  of  Theorem~\ref{PA4}}

We recall
$\pi_f(k\,|a)=\frac{c}{c+f(k,a)}\prod_{i=0}^{k-1}\frac{f(i,a)}{c+f(i,a)}$
and  state  our  weak  law  of  large  numbers.

\begin{lemma}\label{PA41} Suppose $X$ is coloured preferential
attachment random graph with colour law $\mu:\skrix\to(0,1]$ and
linear weight function $f:\skrin\times\skrix^*\to[0,\infty].$ Then,
for  any  $\eps>0),$

$$
\lim_{n\to\infty}\P\Big\{\big|M_X(k,a)-\pi_{f}(k\,|a)\mu\otimes\mu(a)\big|\ge
\eps \Big\}=0$$

and
$$
\lim_{n\to\infty}\P\Big\{\big|\sum_{k=0}^{\infty}M_X(k,a)-\mu\otimes\mu(a)\big|\ge
\eps \Big\}=0.$$

\end{lemma}
\begin{proof}
To begin ,  the proof  of Lemma~\ref{PA41}   we  define  the closed
set   $$F=\Big\{
\omega\in\skrim(\skrin\times\skrix^*):\,\big|\omega(k,a)-\pi_{f}(k\,|a)\mu\otimes\mu(a)\big|\ge
\eps
\,\,\mbox{or}\,\,\big|\sum_{k=0}^{\infty}\omega(k,a)-\mu\otimes\mu(a)\big|\ge
\eps \Big\}.$$
\end{proof}

Notice,  by  Theorem~\ref{PA3} we  have  that

\begin{equation}\label{WLLN}
\limsup_{n\to\infty}\frac{1}{n}\log\P\Big\{M_X\in
F\big\}\le-\inf_{\omega\in F}J(\omega).
\end{equation}

We  end  the  proof  of  the  Lemma by  showing  that  the  left
hand side  of  \eqref{WLLN}  is  negative. For  this  purpose we
suppose  that  there  exists  a  sequence  $\omega_n$  such  that
$J(\omega_n)\downarrow 0.$  Then,  because  $J$  is  good  rate
function and all its level sets  are  compact,  and  by lower
semi-continuity  of  the  mapping $\omega\to J(\omega)$   there is a
limit $\omega\in F$ with $J(\omega)=0.$  Then,  we  have
$H\Big(\omega_{2,1}\,\|\,\mu\Big)=0$  and
$$\sum_{a\in\skrix}\omega_{2}(a)H\Big(\omega(\cdot|a)\,\|\,\sfrac{c}{f(\cdot,\,a)}\otimes\hat{\omega}(\cdot|\,a)\Big)=0.$$
This  implies  $\omega_{2,1}(a_1)=\mu(a_1)$  and
$\omega(k\,|a)=\pi_f(k\,|a)$ which contradicts $\omega\in F.$ We
begin by  recalling   the distribution  of  the  typed graph $X$ as
follows
$$\begin{aligned}
\P_{f}^{(n)}(X=x)=\prod_{m=1}^{n}\mu(x(m))\times &
\prod_{m=2}^{n}\frac{f(n^{(m)}(j_m),
\,a(m))}{\sum_{i=1}^{m-1}f(n^{(m)}(i),\,a(i))}.
\end{aligned}
$$
and  note  that,  $$\begin{aligned}-\sfrac{1}{n}\log
P(x)=-\sfrac{1}{n}\log\mu(x(1))-\sfrac{(n-1)}{n}\langle \log \mu,
M_X\rangle -\sfrac{(n-1)}{n}\langle \log
f,M_X\rangle-&\sfrac{1}{n}\sum_{m=2}\log(m-1)\\
&-\sum_{m=2}^{n}\sfrac{m-(m-1)}{n}\log\langle
f,\,L_{\sfrac{m}{n}}^{X}\rangle.
\end{aligned}$$

Now $\sfrac{1}{n}\log\mu(x(1))\to 0,$
$\sfrac{1}{n}\sum_{m=2}\log(m-1)$ converges to  $0$  and
$$\sum_{m=2}^{n}\sfrac{m-(m-1)}{n}\log\langle
f,\,L_{\sfrac{m}{n}}^{X}\rangle\to \int_{0}^{1}\log\langle
f,\,\nu_{t}\rangle dt=\,\log c,$$  as  $n$  approaches  infinity.
Further,
$$\langle \log
\mu,M_X\rangle\to \langle \log \mu,\mu\rangle\,\,\mbox{and}\,\,
\langle \log f,M_X\rangle\to\langle \log
f,\mu\otimes\mu\otimes\pi_{f}\rangle,$$ by  Lemma~\ref{PA41} as $n$
approaches infinity,  which  completes the  proof  of  the  AEP.

\textbf{Acknowledgements}

We  are  thankful  to  the  referees  for their suggestions which
have helped  improved  this article.

\pagebreak

\textbf{References}

\hangafter=1

\setlength{\hangindent}{2em}

 {\sc  Barab´asi, A.} and
{\sc  Albert, R.}(1999).
\newblock {Emergence of Scaling in Random Networks.}
\newblock{http://arxiv.org/pdf/cond-mat/9910332v1.pdf.}
\smallskip

\hangafter=1

\setlength{\hangindent}{2em}

{\sc  Biggins, J.D.}(2004).
\newblock{Large deviations for mixtures.}
\newblock{\emph{ El. Comm. Probab.9 60 71 (2004).}}
\smallskip

\hangafter=1
\setlength{\hangindent}{2em} {\sc Choi, J.} and {\sc
Sethuraman, S.}(2011)
\newblock{Large deviations of  the degree structures in
P.A schemes. }
\newblock \emph{The  annals  of  applied  probability, 23, 722-763.}
\smallskip

\hangafter=1 \setlength{\hangindent}{2em} {\sc Dereich, S.} and {\sc
Morters, P.}(2009).
\newblock{Random networks with sublinear preferential  attachement: Degree evolutions.}
\newblock{\emph{Electronic Journal of Probability, 14, pp.
1222-1267.}}
\smallskip

\hangafter=1
\setlength{\hangindent}{2em}

{\sc Doku-Amponsah, K.}(2006)
\newblock{Large deviations and information theory for hierarchical and networked data structures.}
\newblock PhD Thesis, Bath (2006).
\smallskip

\hangafter=1 \setlength{\hangindent}{2em}

 {\sc Doku-Amponsah,
K.}(2006)
\newblock{Asymptotic equipartition properties for hierarchical and networked  structures.}
\newblock ESAIM:\emph{Probability  and Statistics}.DOI: 10.1051/ps/2010016, http://dx.doi.org/10.1051/ps/2010016
\smallskip

\hangafter=1 \setlength{\hangindent}{2em}
{\sc Bryc, W.}, {\sc
~Minda, D.} and {\sc ~Sethuraman, S.}(2009).
\newblock{Large deviations for the leaves in some random trees .}
\newblock { \emph{Adv. in  Appl. Probab. Volume 41, Number 3 (2009),
845-873.}http://dx.doi.org/10.1239/aap/1253281066}
\smallskip

\hangafter=1 \setlength{\hangindent}{2em}

 {\sc Doku-Amponsah, K.}  and
{\sc~M\"orters, P.}(2010).
\newblock{Large deviation principle for  empirical measures of
coloured random graphs.}
\newblock {\emph{The  annals  of  Applied Probability}}, 20,
1989-2021(2010).http://dx.doi.org/10.1214/09-AAP647

\smallskip
\hangafter=1 \setlength{\hangindent}{2em}

 {\sc Dembo,A.},{\sc
M\"orters, P.} and {\sc Sheffield, S.}(2003)
\newblock Large deviations of Markov chains indexed by random trees.
\newblock \emph{Ann. Inst. Henri Poincar\'e: Probab.et Stat.41,}
(2005) 971-996.
\smallskip

\hangafter=1

\setlength{\hangindent}{2em}

 {\sc Dembo, A.} and {\sc
O.~Zeitouni, O.}(1998).
\newblock Large deviations techniques and applications.
\newblock Springer, New York, (1998).
\smallskip

\hangafter=1

\setlength{\hangindent}{2em}
 {\sc Dembo, A.} and {\sc
I. Kontoyiannis, I.}(2002).
\newblock{ Source Coding, Large  deviations  and  Approximate
Pattern.  Invited  paper  in  IEEE  transaction  on information
theory, 48(6):1590-1615,June(2002).}
\smallskip

\hangafter=1 \setlength{\hangindent}{2em}

 {\sc Lawrence, S.} and
{~Giles, C.L.}(1998)(1999).
\newblock {Science 280, 98 (1998); Nature 400, 107 (1999).}
\smallskip

\hangafter=1 \setlength{\hangindent}{2em}

 {\sc Rudas, B.},{\sc Toth,
B. } and {\sc ~Valko, B.}(2008).
\newblock{Random Trees and General Branching Processes.}
\newblock {http://arxiv.org/abs/math/0503728}
\smallskip

\bigskip

\vspace{0.5cm}
\textbf{Copyrights}

Copyright for this article is retained by the author(s), with first publication rights granted to the journal.

This is an open-access article distributed under the terms and conditions of the Creative Commons Attribution license (http://creativecommons.org/licenses/by/3.0/).

\end{document}